\documentclass[prl,twocolumn,english,reprint, longbibliography, superscriptaddress, breaklinks=true, showkeys, showpacs=false, nofootinbib]{revtex4}
\usepackage[T1]{fontenc}
\usepackage[latin9]{inputenc}
\usepackage{color}
\usepackage{babel}
\usepackage{amsmath}
\usepackage{amsthm}
\usepackage{amsfonts}
\usepackage{amssymb}
\usepackage{graphicx}
\usepackage{subfigure}
\usepackage{physics}
\usepackage{appendix}

\newtheorem{teo}{Theorem}

\usepackage[unicode=true,pdfusetitle,
 bookmarks=true,bookmarksnumbered=false,bookmarksopen=false,
 breaklinks=true,pdfborder={0 0 0},backref=false,colorlinks=true]
 {hyperref} 
\usepackage{breakurl}

\makeatletter
\@ifundefined{textcolor}{}
{%
 \definecolor{BLACK}{gray}{0}
 \definecolor{WHITE}{gray}{1}
 \definecolor{RED}{rgb}{1,0,0}
 \definecolor{GREEN}{rgb}{0,1,0}
 \definecolor{BLUE}{rgb}{0,0,1}
 \definecolor{CYAN}{cmyk}{1,0,0,0}
 \definecolor{MAGENTA}{cmyk}{0,1,0,0}
 \definecolor{YELLOW}{cmyk}{0,0,1,0}
}

\pdfoutput=1
\hypersetup{colorlinks=true,citecolor=blue,linkcolor=cyan,urlcolor=blue,filecolor= green, breaklinks=true}
\usepackage{url}
\usepackage{breakurl}

\makeatother

\begin{document}

\title{Entanglement monotones from complementarity relations}

\author{Marcos L. W. Basso}
\email{marcoslwbasso@hotmail.com}
\address{Departamento de F\'isica, Centro de Ci\^encias Naturais e Exatas, Universidade Federal de Santa Maria, Avenida Roraima 1000, Santa Maria, Rio Grande do Sul, 97105-900, Brazil}
\address{New adress: Centro de Ci\^encias Naturais e Humanas, Universidade Federal do ABC, Avenida dos Estados 5001, 09210-580 Santo Andr\'e, S\~ao Paulo, Brazil}

\author{Jonas Maziero}
\email{jonas.maziero@ufsm.br}
\address{Departamento de F\'isica, Centro de Ci\^encias Naturais e Exatas, Universidade Federal de Santa Maria, Avenida Roraima 1000, Santa Maria, Rio Grande do Sul, 97105-900, Brazil}

\selectlanguage{english}%

\begin{abstract} 
Bohr's complementarity and Schr\"{o}dinger's entanglement are two prominent physical characters of quantum systems. In this article, we formally connect them.
It is known that complementarity relations for wave-particle duality are saturated only for pure, single-quanton, quantum states. For mixed states, the wave-particle quantifiers never saturate a complementarity relation and can even reach zero for a maximally mixed state. To fully characterize a quanton, it is not enough to consider its wave-particle aspect; we have also to regard its quantum correlations with other systems. Here we prove that for any complete complementarity relation involving predictability and visibility measures that satisfy the criteria established in the literature, the corresponding quantum correlations are entanglement monotones. Therefore, we formally connect entanglement monotones with complementarity relations without appealing to a particular measure.

\end{abstract}


\maketitle


Wave-particle duality is one of the quantum features that occupies a central position in the realm of quantum weirdness, turning apart the quantum world from the classical world. This aspect is the best known example of Bohr's complementarity principle, which states that quantum systems, or quantons \cite{Leblond}, possess properties that are equally real but mutually exclusive \cite{Bohr}. The first efforts to quantify the wave-particle duality were made in Refs. \cite{Zurek, Engle, Yasin}, wherein quantitative measures of wave and particle properties were built and constrained in a complementarity inequality: \begin{equation}
    P^2 + V^2 \le 1, \label{eq:cr1}
\end{equation}
where $P$ is the path predictability and $V$ is the interference pattern visibility. However, with the quantitative formulation of the wave-particle duality, it was noticed that these aspects of a quantum system are not necessarily mutually exclusive. An experiment can provide partial information about the wave and the particle nature of a quantum system, as was experimentally verified in Ref. \cite{Auccaise}. Formalizing, for the first time, the quantitative formulation of the wave-particle duality, D\"urr \cite{Durr} and Englert \textit{et al.} \cite{Englert} devised reasonable criteria for checking the reliability of newly defined predictability and visibility measures. Another important step towards the quantification of the quanton's wave aspect was the realization that the quantum coherence \cite{Baumgratz} would be a good generalization for visibility measures \cite{Bera, Bagan,  Mishra}. 

However, complementarity relations like the one in Eq. (\ref{eq:cr1}) are saturated only for pure, single-quanton, quantum states. This kind complementarity inequality does not really catch a balanced exchange between $P$ and $V$ because the inequality permits, for instance, that $V$ decreases due to the interaction of the system with its environment while $P$ can remain unchanged. Or even worse, $P$ can decrease together with the visibility of the system, thus allowing for the extreme case $P = V = 0$. Hence, something must be missing from Eq. (\ref{eq:cr1}). As noticed by Jakob and Bergou \cite{Janos}, this lack of local information about the system is due to entanglement. This means that the information is being shared with another system, and this kind of quantum correlation can be seen as responsible for the loss of purity of each subsystem such that, for pure maximally entangled states, it is not possible to obtain information about the local properties of the subsystems. As showed by these authors, for a particular pair of measures $P$ and $V$, the entanglement concurrence \cite{Wootters} is recognized as the appropriate quantum correlation measure in a bipartite state of two qubits that completes relation (\ref{eq:cr1}). Later, they extended this idea for composite bipartite systems of arbitrary dimension \cite{Jakob, Bergou}, suggesting that there must exist a complementarity relation between the information of the local properties of each subsystem and the entanglement of the composite system. The triality relations for qubits obtained in Ref. \cite{Janos} were recently verified experimentally in Ref. \cite{Qian}. These relations are also known as complete complementarity relations (CCRs). Considering that we can always purify our system of interest and think of it as part of a multipartite pure quantum system, in Ref. \cite{Marcos} we proposed a general formalism where CCRs are obtained by exploring the purity of a multipartite quantum system, thus encompassing and generalizing the CCRs of Refs. \cite{Jakob, Bergou}. Besides, we extended this formalism for bipartite mixed states in Ref. \cite{Wayhs}.

In this article, considering this formalism to obtain CCRs in which the predictability and visibility measures satisfy the criteria established by  D\"urr \cite{Durr} and Englert \textit{et al.} \cite{Englert}, we prove that the associated quantum correlation measures are entanglement monotones for global pure cases, such that the extension for the mixed case can be given through the convex roof procedure. This is a very interesting and general result, since it formally connects entanglement monotones with complementarity relations, without appealing to specific measures. In addition, this result summarizes all the CCRs known in the literature, as well as opens the possibility for establishing new entanglement measures whenever a complementarity relation that satisfies the criteria of Refs. \cite{Durr, Englert} is derived. Besides, by considering complementarity relations derived in Ref. \cite{Maziero}, we obtain four entanglement monotones. Two of them are already well known entanglement measures, i.e., the linear and von Neumann entropies, while the other two were obtained quite recently, in Ref. \cite{Leopoldo}, by exploring the relation between uncertainty and complementarity. It is worth mentioning that these CCR's were recently verified using IBM Quantum Experience quantum computers, where we performed experimental tests of these complementarity relations applied to a particular class of
one-qubit quantum states and also for random quantum states of one, two, and three qubits \cite{Mauro}.


Following Ref. \cite{Zhu}, let us denote $\mathcal{D}(\mathcal{H)}$ the set of density matrices on $\mathcal{H} \simeq \mathbb{C}^d$ and $U(d)$ the group of unitary operators on $\mathcal{H}$. In addition, let $\mathcal{F}_U$ be the set of local unitarily invariant functions on $\mathcal{D}(\mathcal{H)}$ such that each function $f \in \mathcal{F}_U$ is defined on the space of density matrices for each positive integer $d = \dim \mathcal{H}$. For any given $d$, the function satisfies
\begin{align}
    f(U\rho U^{\dagger}) = f(\rho)\ \forall\  \rho \in \mathcal{D}(\mathcal{H)} \text{ and } U \in U(d). \label{eq:unit}
\end{align}
Therefore $f(\rho)$ can be taken as a function of the eigenvalues of $\rho$. By restricting ourselves to the set $\mathcal{F}_{Uc} \subset \mathcal{F}_{U}$ of local unitarily invariant and real concave functions on $\mathcal{D}(\mathcal{H)}$, then each function $f \in \mathcal{F}_{Uc}$ satisfies Eq. (\ref{eq:unit}) and
\begin{align}
    f(\lambda \rho + & (1 - \lambda) \sigma) \ge \lambda f(\rho) + (1 - \lambda) f(\sigma) \nonumber \\
    & \forall \rho, \sigma \in \mathcal{D}(\mathcal{H)}, \lambda \in[0,1],
\end{align}
for any given $d$. Now, let $\mathcal{H} \simeq \mathcal{H}_A \otimes \mathcal{H}_B$ be a bipartite Hilbert space of a bipartite quantum system $A$ and $B$ with dimension $d_A = d_B = d$. The fact that the dimensions of the subsystems are the same is not essential here. Any function $f \in \mathcal{F}_{Uc} $ can be used to construct an entanglement monotone $E_f$ on $\mathcal{D}(\mathcal{H)}$ as follows. For a pure state $\ket{\Psi} \in \mathcal{H}$,
\begin{align}
    E_f(\Psi) := f(\Tr_B(\ketbra{\Psi})) = f(\rho_A). \label{eq:pure}
\end{align}
Then, it is possible to extend the monotone for mixed states $\rho \in \mathcal{D}(\mathcal{H)}$ by the convex roof construction:
\begin{align}
    E_f(\rho) := \min_{\{p_j, \ket{\Psi_j}\}} \sum_j p_j E_f(\Psi_j), \label{eq:mixed}
\end{align}
where the minimization runs over all pure state ensembles $\{p_j, \ket{\Psi_j}\}$ for which $\rho = \sum_j p_j \ketbra{\Psi_j}$. Conversely, the restriction to pure states of any entanglement monotone is identical to $E_f$ for a given $f \in \mathcal{F}_{Uc}$. These facts were stated as a theorem in Ref. \cite{Vidal}.

Besides that, let $\Delta_d$ be the probability simplex of probability vectors with $d$ components. A function on $\Delta_d$ is symmetric if it is invariant under permutations of the components of probability vectors. Let $\mathcal{F}_s$ be the set of symmetric functions on the probability simplex such that each function $f \in \mathcal{F}_s$ is defined for any given positive integer $d$. The authors in Ref. \cite{Zhu} showed that any symmetric function  $f \in \mathcal{F}_s$  can be lifted to an unitarily invariant function on $\mathcal{D}(\mathcal{H)}$: $\hat{f}(\rho) := f(\text{eig}(\rho)), \ \forall \rho \in \mathcal{D}(\mathcal{H)},$ where $\text{eig}(\rho)$ are the eigenvalues of $\rho$. Conversely, any unitarily invariant function $f$ on the space of density matrices defines a symmetric function on the probability simplex when restricted to diagonal density matrices: $\check{f}(p) := f(\rho_{diag}), \ \forall p \in \Delta_d,$ where $p\in \Delta_d$ represents a probability distribution in $\Delta_d$, which in this case is given by $\rho_{diag}$.  Therefore, for any concave function $f \in  \mathcal{F}_s$, $E_f$ defined by Eqs. (\ref{eq:pure}) and (\ref{eq:mixed}) is an entanglement monotone. Conversely, the restriction to pure states of any entanglement monotone is identical to $E_f$ for a given concave function $f \in  \mathcal{F}_s$. This claim was proved in Ref. \cite{Zhu}, and is essential for proving the theorem we give in the sequence.

Before stating the main result of this work, let us consider an example of how it is possible to obtain an entanglement monotone using complementarity relations. In Ref. \cite{Marcos}, we considered a quantum system described by the density operator $\rho_A$ of dimension $d_A$. The relative entropy of coherence of this state is defined as \cite{Baumgratz}
\begin{align}
    C_{re}(\rho_A) = \min_{\iota \in I} S_{vn}(\rho_A||\iota),
\end{align}
where $I$ is the set of all incoherent states, $S_{vn}(\rho_A||\iota) = \Tr(\rho_A \log_2 \rho_A - \rho_A \log_2 \iota)$ is the relative entropy, and $S_{vn}(\rho)$ denotes the von Neumann entropy of $\rho$. The minimization procedure leads to $\iota = \rho_{Adiag} = \sum_{i = 1}^{d_A} \rho^A_{ii} \ketbra{i}$. Thus 
\begin{align}
    C_{re}(\rho_A) = S_{vn}(\rho_{Adiag}) - S_{vn}(\rho_A) \label{eq:cre}.
\end{align}
Since $C_{re}(\rho_A) \le S_{vn}(\rho_{Adiag})$, it is possible to obtain an incomplete complementarity relation from this inequality:
\begin{equation}
    C_{re}(\rho_A) + P_{vn}(\rho_A) \le \log_2 d_A \label{eq:cr6},
\end{equation}
with $P_{vn}(\rho_A) := \log_2 d_A - S_{vn}(\rho_{Adiag}) = \log_2 d_A + \sum_{i = 0}^{d_A - 1} \rho^A_{ii} \log_2 \rho^A_{ii}$ being a good measure of predictability, already defined in Ref. \cite{Maziero}, while $C_{re}(\rho_A)$ is a bone-fide measure of visibility. It is worth mentioning that the possibility of defining $P_{vn}(\rho_A)$ in this way is due to the fact that the diagonal elements of $\rho_A$ can be interpreted as a probability distribution, which is a consequence of the properties of the density matrix $\rho_A$. Since the complementarity relation in Eq. (\ref{eq:cr6}) can be considered incomplete due to the presence of correlations with other systems, then, by considering $\rho_A$ as part of a bipartite pure quantum system $\ket{\Psi}$, it follows that 
\begin{align}
    E_f(\Psi) := \log_2 d_A - P_{vn}(\rho_A) - C_{re}(\rho_A) \label{eq:entmon}
\end{align}
is an entanglement monotone of the form $E_f(\Psi) := f(\Tr_B(\ketbra{\Psi})) = f(\rho_A)$, that can be extended to the mixed case by the convex roof procedure. To see this, it is enough to note that since $P_{vn}(\rho_A)$ and $C_{re}(\rho_A)$ are convex functions of $\rho_A$, then  $f(\rho_A)$ is concave \cite{Roberts}. In addition, the sum $P_{vn}(\rho_A) + C_{re}(\rho_A)$ is invariant under local unitary operations. Besides, an equivalent but normalized definition is given by $ E_f(\Psi) := 1 - \frac{1}{\log_2 d_A}(P_{vn}(\rho_A) + C_{re}(\rho_A))$. Actually, the entanglement monotone defined in Eq. (\ref{eq:entmon}) is a well known measure of entanglement for global pure cases \cite{Vedral}: $E_f(\Psi) = S_{vn}(\rho_A)$, which leads to the following complete complementarity relation
\begin{equation} C_{re}(\rho_A) + P_{vn}(\rho_A) + S_{vn}(\rho_A) = \log_2 d_A \label{eq:cre_}.
\end{equation}
Therefore, we can use any complementarity relation of the type in Eq. (\ref{eq:cr6}) to obtain entanglement monotones $E_f$ defined by Eqs. (\ref{eq:pure}) and (\ref{eq:mixed}), provided that the measures of predictability and visibility satisfy the criteria established in Refs. \cite{Durr, Englert}, and stated in Appendix \ref{sec:appe}.

\begin{teo}
Let $P(\rho_A) + C(\rho_A) \le \alpha$ be a complementarity relation for the state $\rho_A$ such that it saturates only if $\rho_A$ is pure, with $P(\rho_A)$ and $C(\rho_A)$ being bone-fide measures of predictability and visibility, respectively, and $\alpha\in \mathbb{R}$ with $\alpha>0$. The quantity
\begin{equation}
    E_f := \alpha - P(\rho_A) - C(\rho_A) \label{eq:entmon_}
\end{equation}
is an entanglement monotone, as defined by Eqs. (\ref{eq:pure}) and (\ref{eq:mixed}), when restricted to the Schmidt's coefficients.
\end{teo}

\begin{proof}
Since the measures $P(\rho_A)$ and $C(\rho_A)$  satisfy the criteria established in Refs. \cite{Durr, Englert}, then they are convex functions, what implies that $E_f := \alpha - P(\rho_A) - C(\rho_A) = f(\rho_A)$ is a concave function. Now, let $\ket{\Psi} \in \mathcal{H}_A \otimes \mathcal{H}_B$ be a purification of $\rho_A$, i.e., $\rho_A = \Tr_B \ket{\Psi}\bra{\Psi}$. Using the Schmidt decomposition $\ket{\Psi} = \sum_k \sqrt{\lambda_k}\ket{\phi_k}_A \otimes \ket{\psi_k}_B$, we can write $\rho_A =  \sum_k \lambda_k \ketbra{\phi_k}$, which implies that $C(\rho_A) = 0$ and $P(\rho_A) \neq 0$. But $P(\rho_A)$ must be invariant under permutations of the states' indexes, which implies that $P(\rho_A)$ is invariant under the permutations of the components of the probability vectors $\vec{\lambda}=(\lambda_{0},\cdots,\lambda_{d-1})$, and therefore $E_f = f(\rho_A)$ is invariant under the permutations of the components of $\vec{\lambda}$.
\end{proof}

An equivalent but normalized definition of Eq. (\ref{eq:entmon_}) is given by $E_f:= 1 - \frac{1}{\alpha}(P(\rho_A) + C(\rho_A))$. From Eq. (\ref{eq:entmon_}), it is straightforward to see that all states that maximize $E_f$ has the same form, i.e., these states are those in which the reduced states satisfy $P(\rho_A) + C(\rho_A) = 0$. From the criteria in Appendix \ref{sec:appe}, it is easy to see that these reduced states are the ones that are maximally mixed. Now, it is worth emphasizing that, to define an entanglement monotone, we have to consider the complementarity relation restricted to the Schmidt coefficients for the symmetric function, defined by Eq. (\ref{eq:entmon_}), to be able to be lifted to an unitarily invariant function according to Ref. \cite{Zhu}, as discussed before.

To exemplify the power of this theorem, let us consider the following incomplete complementarity relations derived in Ref. \cite{Maziero} by exploring the properties of the density matrix that describes the state of the subsystem A:
\begin{align}
& P_{l_{1}}(\rho_A) + C_{l_{1}}(\rho_A) \le d_A - 1, \label{eq:crl1} \\
&  P_{hs}(\rho_A) + C_{wy}(\rho_A)  \le (d_A - 1)/d_A \label{eq:crwy},\\
& P_{hs}(\rho_A) + C_{hs}(\rho_A)  \le (d_A - 1)/d_A. \label{eq:crhs}
\end{align}
The quantum coherence/visibility measures appearing in these relations are
\begin{align}
 & C_{l_{1}}(\rho_A) := \min_{\iota}||\rho_A-\iota||_{l_{1}} = \sum_{j \neq k} \abs{\rho^A_{jk}}, \\
& C_{wy}(\rho_A)  := \sum_{j}I_{wy}(\rho_A,|j\rangle\langle j|) = \sum_{j\neq k}\abs{\bra{j}\sqrt{\rho_A}\ket{k}}^2, \\
& C_{hs}(\rho_A) := \min_{\iota \in I}||\rho_A-\iota||_{hs}^{2} = \sum_{j \neq k} \abs{\rho^A_{jk}}^2,\label{eq:hsc}
\end{align}
where $I$ is the set of all incoherent states, the Hilbert-Schmidt's norm of a matrix $M\in\mathbb{C}^{d \times d}$ is defined as $\norm{M}_{hs}:=  \sqrt{\sum_{j,k} \abs{M_{jk}}^2}$, whereas the $l_1$-norm is given by $\norm{M}_{l_1} := \sum_{j,k}\abs{M_{jk}}$, $I_{wy}(\rho,\ketbra{j}) = -\frac{1}{2}\Tr ([\sqrt{\rho},\ketbra{j}]^2)$ is the Wigner-Yanase skew information \cite{Yu}, and $\{|j\rangle\}_{j=0}^{d_{A}-1}\equiv\{|j\rangle_{A}\}_{j=0}^{d_{A}-1}$ is a reference basis. The predictability measures are given by
\begin{align}
& P_{l_{1}}(\rho_A) := d_A - 1 - \sum_{j \neq k} \sqrt{\rho^A_{jj} \rho^A_{kk}}, \\
& P_{hs}(\rho_A) := S^{max}_l - S_l(\rho_{A diag}) = \sum_j (\rho^A_{jj})^2 - 1/d_A,
\end{align}
where $S_l(\rho_{A diag}):= 1 - \Tr \rho^2_{A diag}$ is the linear entropy. Moreover, the incomplete complementarity relation describes the local aspects of a quanton, and therefore is closely linked to the purity of the system. However, it is noteworthy that the complementarity relation given by Eq. (\ref{eq:crl1}) is not invariant under unitary transformations, while the others two given by Eqs. (\ref{eq:crwy}) and (\ref{eq:crhs}) are. Now, when restricted to the Schmidt coefficients, from Theorem 1, it follows that 
\begin{align}
& W_{l_1}(\rho_A) := d_A - 1 - P_{l_1}(\rho_A) - C_{l_1}(\rho_A),\\
& W_{wy}(\rho_A) := \frac{d_A - 1}{d_A} - P_{hs}(\rho_A) - C_{wy}(\rho_A),\\
& S_l(\rho_A) := \frac{d_A - 1}{d_A} - P_{hs}(\rho_A) - C_{hs}(\rho_A),
\end{align}
are entanglement monotones for global pure states, since $C_{\tau} = 0, \ \tau = l_1, wy, hs$ and $P_{\tau}$ is invariant under the permutation of the of the probability vectors $\vec{\lambda}=(\lambda_{0},\cdots,\lambda_{d-1})$ given by the Schmidt's coefficients, and therefore, the quantities above are invariant under local unitary transformations according to Ref. \cite{Zhu}. Besides, $S_l(\rho_A)$ is the well known linear entropy that completes the complementarity relation given by Eq. (\ref{eq:crhs}), which was considered in Refs. \cite{Jakob, Marcos, Huber}. In terms of the Schmidt's coefficients, the measures are given by
\begin{align}
& W_{l_1}(\rho_A) := \sum_{j \neq k}\sqrt{\lambda_j \lambda_k}, \label{eq:robu} \\
& W_{wy}(\rho_A) := \sum_j((\sqrt{\lambda_j})^2- \lambda_j^2) = 1 - \sum_j \lambda^2_j,
\end{align}
with $W_{wy}(\rho_A) = S_l(\rho_A)$ if written in terms of the Schmidt's coefficients. Besides, from Eq. (\ref{eq:robu}), it is possible to see that $W_{l_1}(\rho)$ is equal to the robustness of entanglement for global pure states \cite{Tarrach}. To see that $W_{l_1}(\rho_A)$ and $W_{wy}(\rho_A)$ are concave is straightforward, since all measures of predictability and visibility are convex functions. Also, the measures are invariant under the permutation of the components of the probability vectors, even when we are considering a reference basis in which the measures of quantum coherence are non-vanishing, given that we first restrict ourselves to the Schmidt coefficients and its eigenvectors, and then make a unitary transformation to the reference basis. 

However, there is another way to see that $W_{l_1}(\rho_A)$ and $W_{wy}(\rho_A)$ do not increase under LOCC (local quantum operations and classical communication). As the entanglement cannot increase by LOCC , any entanglement monotone $E$ must reverse the majorization order \cite{Olkin}. Since Schur concave functions are exactly the functions that invert the majorization order, i.e.,  $f(x) \ge f(y)$ whenever $x \prec y$, they can be used to quantify entanglement of pure states \cite{Nie}. Thus, the Schur concave criteria together with invariance under the permutation of the
Schmidt coefficients is enough to show that any function of the type of Eq. (\ref{eq:entmon_}) does not increase under LOCC, according to Ref. \cite{Mintert}.

Summing up, in this article we formally connected two notions of upmost importance for Quantum Mechanics: Entanglement and complementarity. We proved that it is possible to obtain entanglement monotones for global pure cases, and extend them for the mixed case through the convex roof procedure, from any complementarity relation using only the hypothesis that the predictability and visibility measures satisfy the criteria established in the literature. This makes the result very general, and summarizes all the complete complementarity relations known in the literature, as well as opens the path for establishing new entanglement measures whenever a complementarity relation that satisfies the criteria in Refs. \cite{Durr, Englert} is derived. 

\begin{acknowledgments}
This work was supported by the Coordena\c{c}\~ao de Aperfei\c{c}oamento de Pessoal de N\'ivel Superior (CAPES), process 88882.427924/2019-01, and by the Instituto Nacional de Ci\^encia e Tecnologia de Informa\c{c}\~ao Qu\^antica (INCT-IQ), process 465469/2014-0.
\end{acknowledgments}

\appendix

\section{Criteria for Predictability/Visibility Measures}
\label{sec:appe}
D\"urr \cite{Durr} and Englert \textit{et al.} \cite{Englert} established criteria that can be taken as a standard for checking for the reliability of newly defined predictability measures $P(\rho)$ and interference pattern visibility quantifiers $V(\rho)$. These required properties can be stated as follows:
\begin{itemize}
\item[C1] $P$ must be a continuous function of the diagonal elements of the density matrix and $V$ must be a continuous function of the elements of the density matrix.
\item[C2] $P$ and $V$ must be invariant under permutations of the base states indexes.
\item[C3] If $\rho_{jj}=1$ for some $j$, then $P$ must reach its maximum value, while $V$ must reach its minimum possible value.
\item[C4] If $\{\rho_{jj}=1/d\}_{j=0}^{d-1}$, then $P$ must reach its minimum value. In addition, if $\rho$ is pure, then $V$ must reach its maximum value.
\item[C5] If $\rho_{jj}>\rho_{kk}$ for some $(j,k)$, the value of $P$ cannot be increased by setting $\rho_{jj}\rightarrow\rho_{jj}-\epsilon$ and $\rho_{kk}\rightarrow\rho_{kk}+\epsilon$, for $\epsilon\in\mathbb{R}_{+}$ and $\epsilon\ll1$. And $V$ cannot be increased when decreasing $|\rho_{jk}|$ by an infinitesimal amount, for $j\ne k$.
\item[C6] $P$ and $V$ must be convex functions, i.e., $f(\sum_i \lambda_i \rho_i)\le \sum_i \lambda_i f(\rho_i)$, with $\sum_i \lambda_i = 1$, $\lambda_i \in [0,1]$, $f = P, V$ and for $\rho_i$ being valid density matrices.
\end{itemize}


\end{document}